\documentclass[11pt,a4paper,reqno]{amsproc}

\input{header2.sty}

\begin{document}


\title[Lower bound on the density of states]{A bound on the averaged spectral shift function\\ and a lower bound on the density of states\\  for random Schr\"odinger operators on \boldmath $\R^d$}

\author[A.\ Dietlein]{Adrian Dietlein}
\address[A.\ Dietlein]{Mathematisches Institut,
  Ludwig-Maximilians-Universit\"at M\"unchen,
  Theresienstra\ss{e} 39,
  80333 M\"unchen, Germany}
\email{dietlein@math.lmu.de}

\author[M.\ Gebert]{Martin Gebert}
\address[M.\ Gebert]{Department of Mathematics, King's College London, Strand, London, WC2R 2LS, UK}
\thanks{M.G. was  supported by the DFG under grant GE 2871/1-1.}
\email{martin.gebert@kcl.ac.uk}

\author[P.\ D.\ Hislop]{Peter D.\ Hislop}
\address[P.\ D.\ Hislop]{Department of Mathematics, University of Kentucky, Lexington, KY 40506-0027, USA}
\email{peter.hislop@uky.edu}
\thanks{P.H. was supported in part by the NSF under grant DMS-1103104.}

\author[A.\ Klein]{Abel Klein}
\address[A.\ Klein]{University of California, Irvine, Department of Mathematics, Irvine, CA 92697-3875, USA}
\thanks{A.K. was  supported in part by the NSF under grant DMS-1001509.}
\email{aklein@uci.edu}

\author[P.\ M\"uller]{Peter M\"uller}
\address[P.\ M\"uller]{Mathematisches Institut,
  Ludwig-Maximilians-Universit\"at M\"unchen,
  Theresienstra\ss{e} 39,
  80333 M\"unchen, Germany}
\email{mueller@lmu.de}

\begin{abstract}
We obtain a  bound
on the expectation of the spectral shift function for alloy-type random Schr\"odinger
operators on $\R^d$ \emph{in the region of localisation}, corresponding to a change from Dirichlet to
Neumann boundary conditions along the boundary of a finite volume. The bound scales with the area of the surface where the boundary conditions are changed.
As an application of our bound on the spectral shift function, we  prove a reverse Wegner inequality
for finite-volume Schr\"odinger operators  in the region of localisation  with a constant
locally uniform in the energy. The application requires that the single-site distribution of the independent and identically distributed random variables has a Lebesgue density that is also bounded away from zero.
The reverse Wegner inequality implies a strictly positive, locally uniform lower bound on the density of states for these continuum random Schr\"odinger operators.
\end{abstract}

\maketitle


\section{Introduction}
\label{sec:intro}

The effect of changing boundary conditions on spectral and scattering properties of Schr\"odinger operators on multi-dimensional Euclidean space $\R^d$ is an important issue that is still far from being well understood.
 This is true even when switching between Dirichlet and Neumann boundary conditions.
 In the \emph{discrete case} the  change from Dirichlet to Neumann boundary conditions along the boundary of a finite volume region is conveyed by a finite-rank operator.  The rank is proportional to the surface area of the region where the boundary conditions are  changed. This results in a bound on the corresponding spectral shift function proportional to this surface area.
For multi-dimensional continuum Schr\"odinger operators, however, this is a delicate issue, as was already pointed out in \cite{artSSF1987Kir} decades ago. Here, the change in  boundary conditions is not given by a finite-rank operator, and there is no uniform bound on the corresponding  spectral shift function of the order of the surface area where the boundary conditions are changed \cite{Nakamura:2001jh, Mine:2002tl}.   Even worse, \cite{artSSF1987Kir} considers the spectral shift function for the Dirichlet Laplacians on a cube of side length $L$ with different boundary conditions on a smaller cube of side length $\ell$ that is inscribed into the big cube.
 Keeping the smaller cube of size $\ell$ fixed, Kirsch proves that the effect of changing  boundary conditions on its surface causes the spectral shift function to diverge in the limit of large $L$.
We refer to Remarks~\ref{ssf-problems} below for more  details.

One of the main results of this  paper is that in the presence of disorder-induced localisation, this effect no longer occurs and the multi-dimensional continuum situation again resembles the discrete version.
In Theorems~\ref{cor:SSF2}
 and \ref{lem:SSFEstimate}  we obtain new bounds
on the expectation of the spectral shift function for alloy-type random Schr\"odinger
operators on $\R^d$  for energies \emph{in the region of localisation}, corresponding to a change from Dirichlet to
Neumann boundary conditions along the boundary of a finite volume.  
Localisation is  the crucial new ingredient in controlling the change of boundary conditions for continuum Schr\"odinger operators.  In the region of localisation a change of boundary conditions is mostly felt near the relevant boundary, an argument we make rigorous for proving   Theorems~\ref{cor:SSF2} and \ref{lem:SSFEstimate}.

As an application of our spectral shift bound, we derive a strictly positive \emph{lower} bound for the density of states of alloy-type random Schr\"odinger operators  on $\R^d$. Such  bounds have been proven for its discrete counterpart,  the Anderson model on $\Zd$  \cite{Jeske92,artRSO2008HiMu},  following an argument given in \cite{Wegner}. Extending the proof of a lower bound on the density of states for the discrete Anderson model  to the case of  multi-dimensional continuum random Schr\"odinger operators requires a sufficiently detailed control of the spectral shift function for a Dirichlet-Neumann bracketing argument.
For this reason, the problem remained open for quite some time. Theorem~\ref{lem:SSFEstimate} is the crucial ingredient in our proof of a reverse Wegner inequality in Theorem~\ref{th:LowerWegner}. From that we deduce
a lower bound for the   density of states of alloy-type random Schr\"odinger operators in  Corollary~\ref{cor:LowerBoundDOS}.

\section{Model and results}

We consider a random Schr\"odinger operator with an alloy-type random potential
\begin{equation}
\label{eq:TheOperator}
\omega \mapsto H_{\omega} := H_0+V_\omega := H_0 + \sum_{k\in\Zd} \omega_k u_k
\end{equation}
acting on a dense domain in the Hilbert space $L^{2}(\Rn)$ for $d\in\N$. Here $H_0$ is a non-random self-adjoint operator and $\omega \mapsto V_{\omega}$ is a random potential subject to the following assumptions.
\be{MyDescription}
\item[(K)]{ The unperturbed operator is given by $H_0 := -\Delta+V_0$ with $-\Delta$ being the
	non-negative Laplacian on $d$-dimensional Euclidean space $\Rn$ and $V_0 \in L^{\infty}(\R^{d})$  is a deterministic, $\Zd$-periodic and bounded
	background potential.  \label{assK}}
\item[(V1)]{ The family of random coupling constants $\omega:=(\omega_k)_{k\in\Zd} \in \R^{\Z^{d}}$ is
	distributed identically and independently according to the Borel probability measure $\PP := \bigotimes_{\Z^{d}}P_{0}$
	on $\R^{\Z^{d}}$.	We write $\E$ for the corresponding expectation. The single-site distribution $P_0$ is absolutely continuous with respect to Lebesgue measure on $\R$. The corresponding Lebesgue density
 $\rho$ is bounded and has support $\supp(\rho) \subseteq [0,1]$.
 	\label{assV1}  }
\item[(V2)]{ The single-site potentials $u_k(\,\cdot\,):=u(\,\cdot\,-k)$, $k\in\Zd$, are translates
	of a non-negative bounded function $0 \le u\in L^\infty_{c}(\R^d)$ with support contained in a ball of radius
	$R_u>0$. There exist constants $C_{u,-}, C_{u,+}>0$ such that
	\begin{equation}
		\label{uk-bounds}
		0< C_{u,-} \leq \sum_{k\in\Zd} u_k \leq C_{u,+}<\infty.
	\end{equation}
	\label{assV2}
	}
\end{MyDescription}

	The above assumptions are not optimal but are chosen in
	order to avoid unnecessary technical complications.
We note that the condition $\supp(\rho) \subseteq [0,1]$ in (V1) is not stronger than the seemingly weaker property $\supp(\rho)$ is compact. In fact, the former can be obtained from the latter
via the inclusion of an additional periodic potential, a change of variables of the random couplings $(\omega_k)_{k\in\Z^{d}}$ and by rescaling the single-site potential $u$.	The random potential
	$V$ need not even be of the precise form \eqref{eq:TheOperator}, as $\Z^{d}$-translation
	invariance can be dropped for most of the arguments that do not involve the IDS or require deterministic spectrum.

The above model is $\Zd$-ergodic with respect to lattice translations. It follows that there exists a closed set $\Sigma\subset\R$, the non-random spectrum of $H$, such that $\Sigma=\sigma(H)$ holds $\mathbb{P}$-almost surely \cite{pastfig1992random}. We drop the subscript $\omega$ from $H$ and other quantities when we think of these quantities as random variables (as opposed to their particular realizations).
The covering conditions \eqref{uk-bounds} imply \cite{pastfig1992random}
\begin{equation}
 	\Sigma_{0} + [0, C_{u,-}] \subseteq \Sigma \subseteq \Sigma_{0} + [0, C_{u,+}],
\end{equation}
where $\Sigma_{0} := \sigma(H_{0})$ is the spectrum of the unperturbed periodic operator.

Given an open subset $G\subset\Rn$, we write $H_G$ for the Dirichlet restriction of $H$ to $G$.
We define the random finite-volume eigenvalue counting function
\beq
 \R \ni E \mapsto N_L(E):= \Tr \left(\id_{(-\infty,E]}(H_L)\right)
\eeq
for $L>0$, where $\id_{B}$ stands for the indicator function of a set $B$, $H_L:=H_{\Lambda_L}$ and $\Lambda_L:=(-L/2,L/2)^{d}$ for the open cube about the origin of side-length $L$. The Wegner estimate holds under our assumptions: given a bounded interval $I\subset\R$ and $E_1,E_2\in I$ with $E_1<E_2$, we have
\beq
\label{eq:WegnerEst1}
\mathbb{E}\big[N_L(E_2) - N_L(E_1)\big] \le C_{W,+}(I) |\Lambda_L| (E_2-E_1)
\eeq
for all $L>0$, where $C_{W,+}(I)$ is a constant which is polynomially bounded in $I_{+}:= \sup I$, and
$|B|$ is the Lebesgue measure of a Borel-measurable set $B\subseteq \Rn$.
We refer to \cite{MR2362242,MR3046996,MR3106507} for recent developments concerning the Wegner estimate.
Ergodicity implies that, almost surely, the limit
\beq
\label{eq:DOSFinVol}
N(E):= \lim_{L\to\infty} \frac 1 {|\Lambda_L|} N_L(E)
\eeq
exists for all $E\in\R$ in our situation \cite{pastfig1992random}. The non-random limit function $N$ is called the integrated density of states (IDOS) of $H$, see e.g.\ \cite{MR2307751, MR2378428} for reviews. We conclude from the Wegner estimate \eqref{eq:WegnerEst1} that the IDOS $N$ is Lipschitz continuous, hence absolutely continuous with a bounded Lebesgue density $n$. The latter is referred to as the density of states (DOS) of $H$. The Wegner bound for the DOS reads
\begin{equation}
\label{eq:WegnerEst2}
\esssup_{E\in I} n(E) \leq C_{W,+}(I).
\end{equation}
Such upper bounds for the IDOS or DOS have been studied extensively, as they are an important ingredient for most proofs of Anderson localisation via the multi-scale analysis.

One goal of this paper is to derive  lower bounds for the IDOS and DOS of alloy-type random Schr\"odinger operators that complement \eqref{eq:WegnerEst1} and \eqref{eq:WegnerEst2}, respectively.
In the discrete case, i.e., for the classical Anderson model on the lattice $\Zd$, such bounds have been proven in \cite{Jeske92,artRSO2008HiMu}, following an argument given in \cite{Wegner}. The proof can be adapted to apply also to one-dimensional continuum random Schr\"odinger operators.
Even though it is well-known that $\Sigma = \supp(n)$,
and thus $n>0$ Lebesgue-almost everywhere on $\Sigma$,
it is of interest to have a locally uniform lower bound for the DOS. Here, we only mention that the DOS occurs as the intensity of the Poisson point process describing level statistics of eigenvalues in the localised regime.
This is well known by now for the discrete Anderson model \cite{MR1385082, MR2505733} and for a one-dimensional continuum model \cite{MR603503}.  It is likely to be true for multi-dimensional continuum models as well \cite{MR2663411}.

The proof of a lower bound on the DOS for alloy-type random Schr\"odinger operators on $\R^d$, which are  continuum random Schr\"odinger operators,
 requires a sufficiently detailed control of the spectral shift function for a Dirichlet-Neumann bracketing argument. This control is the main accomplishment of this paper. In Theorems~\ref{cor:SSF2}
 and \ref{lem:SSFEstimate}  we obtain new bounds
on the expectation of the spectral shift function for alloy-type random Schr\"odinger
operators on $\R^d$ \emph{in the region of localisation}, corresponding to a change from Dirichlet to
Neumann boundary conditions along the boundary of a finite volume.  Here localisation is the new ingredient that allows us to overcome
 the inherent differences between the lattice model on $\mathbb{Z}^d$ and the continuum model on $\R^d$.
 Theorem~\ref{lem:SSFEstimate} is the crucial ingredient in our proof of a lower bound for the IDOS and DOS of alloy-type random Schr\"odinger operators. These bounds are stated in Theorem~\ref{th:LowerWegner} and Corollary~\ref{cor:LowerBoundDOS}.

We characterise the energy region of complete localisation, see e.g.\ \cite{MR2203782}, for random Schr\"odinger operators in terms of fractional moment bounds \cite{artRSO2006AizEtAl2}.
Let $\chi_x:=\id_{\Lambda_{1}(x)}$ denote the multiplication operator corresponding to the indicator function of the unit cube $\Lambda_{1}(x):=x + \Lambda_{1}$, centred at $x\in\Rn$. Given an open subset $G \subseteq \Rn$, we use the notation $R_z(H_G):=(H_{G}-z)^{-1}$ for the resolvent of $H_G$ with $z\in\mathbb{C}\setminus \sigma(H_{G})$ in the resolvent set of $H_{G}$.

\begin{definition}[Fractional moment bounds]
	\label{DefFMB}
	We write
	 $E\in \FMB:=\FMB(H)$, the region of complete localisation, if there exists a neighbourhood $U_E$ of $E$, a fraction $0<s<1$ and constants
	 $C,\mu>0$ such that for every open subset $G\subseteq \Rn$ and $x,y \in G$ we have the bound
	\begin{equation}
	\label{eq:DefFMB}
	\sup_{E'\in U_E,\eta\neq 0} \mathbb{E}\left[\|\chi_x R_{E'+i\eta}(H_G) \chi_y\|^s \right] \leq C e^{-\mu|x-y|}\,.
	\end{equation}
\end{definition}

\begin{remarks}
\item
	\label{rem:alls}
	If \eqref{eq:DefFMB} holds for \emph{some} $0 < s < 1$, then it holds for \emph{all} $0 < s < 1$
	with constants $C$ and $\mu$ depending on $s$, see \cite[Lemma B.2]{artRSO2001AizEtAl}, which generalises to continuum random Schr\"odinger operators \cite[Lemma~A.2]{DiGeMu16b}.
	In particular, if $I\subset\FMB$ is a compact energy interval, then, for every $0<s<1$ there exist constants $C,\mu$
	such that \eqref{eq:DefFMB} holds with these constants, uniformly in $E\in I$.
\item
	Our proofs do not require the validity of \eqref{eq:DefFMB} for every open subset
	$G\subseteq\Rn$. All we need are subsets that are cubes or differences of cubes.
\item
	Bounds of the form \eqref{eq:DefFMB} have first been derived for the lattice Anderson model in
	\cite{artRSO1993AizMol}, see also \cite{artRSO1998AizGr,artRSO2001AizEtAl}, either for sufficiently strong
	disorder or in the Lifshitz tail regime. They were generalised to continuum random Schr\"odinger operators in
	\cite{artRSO2006AizEtAl2}. The formulation there differs with respect to the distance function that is used.
	We refer to \cite[(8) in App.\ A]{artRSO2006AizEtAl2} for an interpretation. Bounds as in \eqref{eq:DefFMB}
	have been derived in \cite{MR2303305} by an adaptation of the methods from \cite{artRSO2006AizEtAl2}
	in the fluctuation boundary regime.
\end{remarks}


\subsection{Lower bound on the DOS}
\label{sec:lobo}

The validity of our main application rests on an additional assumption on the model.
\be{MyDescription}
\item[(V1')]{ The single-site probability density is bounded away from zero on the unit interval
	\begin{equation}
	  \rho_-:=\essinf_{\nu\in[0,1]}\rho(\nu) > 0.
	\end{equation}}
\end{MyDescription}
In fact, $\rho$ need not be bounded away from zero uniformly
on all of its support; a small neighbourhood of the endpoints $0$ and $1$ could be omitted. 
For simplicity, we will assume (V1') as stated.

As described in \eqref{eq:WegnerEst1}, the usual Wegner estimate is an upper bound on the expectation of the eigenvalue counting function. The following new result is a lower bound on the same quantity that we refer to as a \emph{reverse Wegner estimate}. We use the notation $\Int(A)$ for the interior of a set $A\subset \R$.

\begin{theorem}
\label{th:LowerWegner}
Assume {\upshape(K)}, {\upshape(V1)}, {\upshape(V1')} and {\upshape(V2)}. Consider a compact energy interval $I\subset \FMB\cap \Int\big(\Sigma_0+ [0,C_{u,-}]\big)$. Then there exists a constant $C_{W,-}(I)>0$ and an initial length scale $L_0>0 $ such that
\beq
\label{eq:LowerWegnerStat}
\mathbb{E}\left[ \Tr\big(\id_{[E_1,E_2]}(H_L)\big) \right] \geq C_{W,-}(I) (E_2-E_1) |\Lambda_L|
\eeq
holds for all $E_1,E_2\in I$ with $E_1<E_2$ and all $L> L_0$.
\end{theorem}

\begin{remark}
	\label{rem:everywhere}
	Our proof uses Wegner's original trick \cite{Wegner} which turns the disorder average on the left-hand side of  \eqref{eq:LowerWegnerStat} into an effective shift for the IDOS of the unperturbed operator $H_{0}$. (Here, the covering condition enters in the continuum model.) This is why, in general, we cannot establish the lower bound \eqref{eq:LowerWegnerStat} for all energies in the region of complete localisation $\FMB$.
\end{remark}

It is well known that a Wegner estimate implies existence and boundedness of the DOS. In the same vein,
the reverse Wegner estimate \eqref{eq:LowerWegnerStat} implies a local lower bound on the DOS.
Thus, the next corollary follows from Theorem~\ref{th:LowerWegner} in the same way as \eqref{eq:WegnerEst2} follows from \eqref{eq:WegnerEst1}.

\begin{corollary}
	\label{cor:LowerBoundDOS}
Assume {\upshape(K)}, {\upshape(V1)}, {\upshape(V1')} and {\upshape(V2)}. Consider a compact energy interval $I\subset \FMB\cap \Int\big(\Sigma_0+ [0,C_{u,-}]\big)$.	 Then there exists a constant
	$C_{W,-}(I)>0$ such that
	\beq
		\label{eq:LowerBoundDOS}
		\essinf_{E\in I} n(E) \geq C_{W,-}(I).
	\eeq
\end{corollary}

As was already mentioned, the essential difficulty to overcome in the proof of Theorem \ref{th:LowerWegner} is to estimate the error arising from a local change of boundary conditions. This is why Theorem \ref{th:LowerWegner} is limited to the region of complete localisation.

\subsection{Bounds on the spectral shift function}
 For $L\in\mathbb{R}_{>0}$ we write $H^{D}_{L}$, respectively $H^{N}_{L}$ for the restrictions of the operator $H$ to
$\Lambda_{L}$ with Dirichlet, respectively Neumann, boundary conditions. Moreover, we define the spectral shift function (SSF) of the pair $H^{D}_{L}$ and $H^{N}_{L}$ at energy $E\in\mathbb{R}$ by
\begin{equation}
\label{eq:DefSSF}
\xi\big(E,H^{N}_{L},H^{D}_{L}\big) := \Tr \Big(\id_{(-\infty,E]}(H^{N}_{L})-\id_{(-\infty,E]}(H^{D}_{L}) \Big) \quad (\, \ge 0).
\end{equation}
Note that, as $|\Lambda_L|$ is finite, this definition makes sense and coincides with the abstract definition of the SSF \cite{yafaev1992mathscattering}.
Our  new technical result establishes a local bound for the disorder-averaged SSF inside the region of complete localisation.

\begin{theorem}
\label{cor:SSF2}
Assume {\upshape(K)}, {\upshape(V1)} and {\upshape(V2)}. Given a compact energy interval $I\subset \FMB$ and an arbitrary length $L_{0}$>0,
	there exists a constant $C>0$ such that
	\begin{equation}
		\sup_{E\in I}\mathbb{E}\big[\xi\big(E,H^{N}_{L},H^{D}_{L}\big)\big] \leq C L^{d-1}
	\end{equation}
	holds for all $L >L_{0}$, where $H^{N}_{L}$, respectively $H^{D}_{L}$ are the Hamiltonians in the cube $\Lambda_L $ with Neumann,
	respectively Dirichlet boundary conditions.
\e{theorem}

In order to prove Theorem \ref{th:LowerWegner} we need a modified version of the above
theorem.
For $L,l\in\mathbb{R}_{>0}$ and $x_{0} \in \Lambda_{L}$ such that $\overline{\Lambda_{l}(x_{0})}\subset\Lambda_L$ we write $H^{D}_{L,l}$, respectively $H^{N}_{L,l}$ for the restrictions of the operator $H$ to
$\Lambda_{L}\setminus\overline{\Lambda_l(x_{0})}$ with Dirichlet, respectively Neumann, boundary conditions along the inner boundary $\partial \Lambda_{l}(x_{0})$ and Dirichlet boundary conditions along the outer boundary $\partial\Lambda_{L}$. Then we have a similar bound for the corresponding spectral shift function as in Theorem \ref{cor:SSF2}.

\begin{theorem}
	\label{lem:SSFEstimate}
Assume {\upshape(K)}, {\upshape(V1)} and {\upshape(V2)}. Given a compact energy interval $I\subset \FMB$, there exists a constant $C>0$ such that
	\begin{equation}
	\label{eq:SSFEstStat}
	\sup_{E\in I}\mathbb{E}\big[\xi\big(E,H^{N}_{L,l},H^{D}_{L,l}\big)\big] \leq C l^{d-1}
	\end{equation}
	holds for all $L,l >0$ and $x_{0} \in \Lambda_{L}$, provided $\Lambda_l(x_{0})\subset \Lambda_L$ with $\dist\big(\partial \Lambda_l(x_{0}), \partial \Lambda_L\big)\geq 3$.
\end{theorem}

The proofs of Theorem \ref{cor:SSF2} and  Theorem \ref{lem:SSFEstimate} proceed along the same lines.  Since the domain in   the second theorem is slightly more uncommon, and in addition  the estimate is uniform  in the side length $L$ of the bigger box, we give the proof of  Theorem \ref{lem:SSFEstimate}.
Theorems \ref{cor:SSF2} and   \ref{lem:SSFEstimate} used together provide an analogous estimate for the averaged SSF in the case where the perturbation consists of an \emph{additional} Dirichlet or Neumann boundary on the surface of a smaller subcube.

\begin{corollary}
\label{cor:SSF}
Assume {\upshape(K)}, {\upshape(V1)} and {\upshape(V2)}. Given a compact energy interval $I\subset \FMB$, there exists a constant $C>0$ such that
\begin{equation}
	\label{eq:SSFEstStat2}
	\sup_{E\in I}\mathbb{E}\big[\big|\xi\big(E,H_{L},H^{\star}_{L,l} \oplus H^\star_{\Lambda_{l}(x_{0})}\big)\big|\big] \leq C l^{d-1}
	\end{equation}
	holds for all $L,l >0$ and $x_{0} \in \Lambda_{L}$, provided $\Lambda_l(x_{0})\subset \Lambda_L$ with $\dist\big(\partial \Lambda_l(x_{0}), \partial \Lambda_L\big)\geq 3$. Here, $\star\in \{N,D\}$ denotes Dirichlet, respectively Neumann, boundary conditions on $\partial \Lambda_l(x_{0})$.
\e{corollary}

\be{remarks}
\label{ssf-problems}
\item
	The crucial point in Theorem \ref{lem:SSFEstimate} besides the uniformity of the bound in $L$ is
	the quantitative control in terms of the ``size'' of the perturbation, which, in this case, is
	the volume $|\partial \Lambda_l(x_{0})|$ of the surface where the boundary condition is changed.
	Both properties are needed together in the application for Theorem~\ref{th:LowerWegner}. The spectral
	shift estimate in \cite{MR2352262}, which holds for a particular potential perturbation, is valid
	for Lebesgue-almost all energies. It is uniform in $L$ but provides no control on the ``size'' of the
	perturbation. Known $L^{p}$-bounds \cite{MR1824200,MR1945282,MR2200269} are not as detailed either.
\item
	A result analogous to Corollary \ref{cor:SSF} holds even for arbitrary open subsets $G \subseteq\R^{d}$
	instead of the cube $\Lambda_{L}$, if one considers perturbations by compactly supported, bounded potentials
	instead of an additional boundary condition. This follows from \cite[Thm.~3.1]{DiGeMu16b}.
	However, for the case of an additional boundary condition, as we consider here, some regularity
	of $\partial G$ seems to be needed for the a priori trace-class estimates in the Appendix to extend properly,
	see Lemma \ref{lem:aPriori}.
\item
	The condition $\dist\big(\partial \Lambda_l(x_{0}), \partial \Lambda_L\big)\geq 3$ in
	Theorem~\ref{lem:SSFEstimate} and Corollary~\ref{cor:SSF}   can
	be replaced by $\dist\big(\partial \Lambda_l(x_{0}), \partial \Lambda_L\big)\geq \delta >0$.
	The constant $C$ in the statements would then depend on $\delta$. We took $\delta=3$ for technical convenience.
\item
	The statement of  Theorem~\ref{cor:SSF2}, which is locally uniform in energy,
	should be compared to known estimates on the effect of changing boundary
	conditions in the continuum. For deterministic Schr\"odinger operators with magnetic fields, a
	$L^{1}$-bound for the spectral shift function, which is proportional to $L^{d-1}$,
	can be found in \cite[Thm.~4, Prop.~5]{Nakamura:2001jh}.  The estimates in
	\cite[Thm.\ 6.2]{Doi:2001hy} and \cite[Thm.\ 1.4]{Mine:2002tl} are pointwise in energy, but
	compare the integrated densities
	of states at slightly shifted energies, which introduces an arbitrarily small error of size $L^{d}$.
\item
	In the lattice case and for $d=1$ in the continuum, one can bound the SSF
	$\xi\big(E,H^{D}_{L,l},H^{N}_{L,l}\big)$
	simply by the rank of the perturbation $H^{D}_{L,l}-H^{N}_{L,l}$, which is of order $l^{d-1}$
	independently of $L$.
	For multi-dimensional continuum Schr\"odinger operators, however, it is a subtle problem to obtain
	bounds on the SSF,
	which hold pointwise in energy and are uniform in the volume. This has been noted many times after
	\cite{artSSF1987Kir} had discovered that in $d\geq 2$ for every fixed $l>0$ and every  $E>0$,
	\beq
		\label{kirsch-div}
		\sup_{L>l} \xi\big(E, -\Delta^N_{L,l},-\Delta^D_{L.l}\big) =\infty.
	\eeq
	The divergence in \eqref{kirsch-div} is attributed to the increasing degeneracies of the eigenvalues of the Laplacian for
	larger volumes at fixed energy. This represents a principal danger that would have to be ruled out for
	general multidimensional continuum Schr\"odinger operators. Therefore it is plausible that localisation and the
	disorder average help in this situation and enable us to prove bounds of the form \eqref{eq:SSFEstStat}.
\end{remarks}

The remainder of the paper is organised as follows: the next section contains a proof of Theorem \ref{th:LowerWegner}.
The main novelty, Theorem~\ref{lem:SSFEstimate}, is proven separately in Section \ref{sec:SSFEst}.
There we make repeated use of some a priori trace-class estimates.
We review them in the Appendix for the convenience of the reader.

\section{Proof of Theorem \ref{th:LowerWegner}}
\label{sec:LowerWegner}

In this section we prove Theorem \ref{th:LowerWegner}, given Theorem~\ref{lem:SSFEstimate}.

\subsection{Performing a change of variables.}
\label{subsec:IDOSChangeOfVar}

We fix $E\in\R$ and $\varepsilon>0$.
Let $f_\epsilon\in C^\infty(\R)$ be a smooth, monotone increasing switch function such that
\beq
f_\epsilon(x):=
\begin{cases}
0&\quad x \le 0\\
1&\quad x \ge \epsilon
\end{cases},
\eeq
and its derivative satisfies $0\leq f_\epsilon'\leq (2/\epsilon) \id_{(0,\epsilon)}$. We define its translate $f_{E,\epsilon}:=f_\epsilon(\,\pmb\cdot\,-E)$
and estimate, using the bounds on its derivative and the upper bound from \eqref{uk-bounds},
\begin{align}
\label{eq:DOSBound1}
\frac{1}{L^d\epsilon} \;\E \left[N_L(E+\epsilon)-N_L(E)\right]
&=
\frac{1}{L^d\epsilon} \;\E\l[ \Tr\l(\id_{(E,E+\eps]}(H_L)\r) \r] \notag\\
& \geq
\frac{1}{2L^d} \; \E\big[ \Tr\big( f'_{E,\eps}(H_L) \big) \big] \notag \\
&\ge
\frac{1}{2L^d C_{u,+}} \sum_{k\in\Z^d} \E\big[ \Tr \big(u_k f'_{E,\eps}(H_L)\big) \big].
\end{align}
The $k$-sum in \eqref{eq:DOSBound1} effectively runs over only finitely many $k$, and we will perform it in two steps. To this end we introduce the notation
\begin{equation}
	\label{eq:def-raute}
 	B^{\#}:= B \cap \Z^{d}
\end{equation}
for the set of lattice points in a subset $B\subseteq \Rn$.
Apart from a boundary layer, we partition the cube $\Lambda_{L}$ into smaller cubes
$\Lambda_{l,j}:=\Lambda_l(j)$ of side-length $l\in\N$, $l<L-2R_{u}-6$, centred at
\beq
j \in \Gamma_L^l := \big\{ (k_1,...,k_d)\in (l\Z)^d : |k|_{\infty} \le (L-l)/2-R_u  -4\big\}.
\eeq
Here, $|\pmb\cdot|_{\infty}$ denotes the maximum norm on $\Rn$ and, as in Assumption (V2), the single-site potential $u$
 has support in a ball of radius $R_u>0$.
As $u_{k} \ge 0$ and $f'_{E,\varepsilon} \ge 0$, we infer from \eqref{eq:DOSBound1} that
\begin{equation}
	\label{eq:DOSBound20}
 	\frac{1}{L^d\epsilon} \;\E \left[N_L(E+\epsilon)-N_L(E)\right]
 	\ge
	\frac{1}{2 L^d C_{u,+}} \sum_{j\in \Gamma_L^l}\sum_{k\in\Lambda_{l,j}^{\#}} \mathbb{E}\big[ \Tr\big(u_k f'_{E,	\epsilon}(H_L)\big) \big].
\end{equation}

For $j\in\Gamma_L^l$, we abbreviate
\begin{equation}
\label{eq:DefFl}
F_{\Lambda_{l,j}}(\omega):= \sum_{k\in\Lambda_{l,j}^{\#}} \Tr\big( u_k f'_{E,\epsilon}(H_{\omega,L}) \big),
\end{equation}
where the dependence on the disorder realisation $\omega$ is stressed. We proceed by estimating the expectation
$\E[F_{\Lambda_{l,j}}]$ from below. We denote by $\omega_{\Lambda_{l,j}}$, respectively $\omega_{\Lambda_{l,j}^c}$, the collection of random variables corresponding to single-site potentials centred inside, respectively outside, the cube $\Lambda_{l,j}$. We remark that the function $F_{\Lambda_{l,j}}$ may depend on coupling constants
$\omega_k$ for $k\notin \Lambda_L$ with $|k|_{\infty} < L/2 + R_u$. Assumption~\ref{assV1} implies
\begin{equation}
\label{eq:Est1Fl}
\E[F_{\Lambda_{l,j}}]
\geq
\rho_-^{\theta(l)}\; \mathbb{E}_{\Lambda_{l,j}^{c}}\bigg[ \int_{[0,1]^{\theta(l)}}d\omega_{\Lambda_{l,j}}\, F_{\Lambda_{l,j}}\big( (\omega_{\Lambda_{l,j}},\omega_{\Lambda_{l,j}^c})\big) \bigg],
\end{equation}
where $\mathbb{E}_{\Lambda_{l,j}^{c}}\left[\, \cdot\, \right]$ denotes the expectation with respect to the random variables $\omega_{\Lambda_{l,j}^c}$ and $\theta(l):= |\Lambda_{l,j}^{\#}|$ denotes the cardinality of $\Lambda_{l,j}^{\#}$ which is independent of $j\in\Gamma_L^l$ and of order $l^{d}$.
For each fixed $j$,
we perform the same change of variables as in \cite{Wegner,artRSO2008HiMu}
\beq
\omega_{\Lambda_{l,j}} = \{\omega_{k}\}_{k\in \Lambda_{l,j}^{\#}}
\mapsto
\eta := \{\eta_k\}_{k\in \Lambda_{l,j}^{\#}},
\eeq
$\eta_k:=\omega_k-\omega_{j}$ for $k\in\Lambda_{l,j}^{\#} \setminus \{j\}$ and $\eta_{j} :=\omega_{j}$.
This yields
\begin{align}
	\label{eq:Est2Fl}
	\int_{[0,1]^{\theta(l)}} & d\omega_{\Lambda_{l,j}}\, F_{\Lambda_{l,j}}\big( (\omega_{\Lambda_{l,j}},
			\omega_{\Lambda_{l,j}^c})\big) \notag\\	
	&= \int_{[0,1]}d\eta_{j} \int_{[-\eta_{j},1-\eta_{j}]^{\theta(l)-1}}
		\Bigg(\prod_{\substack{k\in\Lambda_{l,j}^{\#}\\ k\neq j}}\!\! d\eta_k \Bigg)\, F_{\Lambda_{l,j}}
		\big( (\omega_{\Lambda_{l,j}}(\eta),\omega_{\Lambda_{l,j}^c})\big) \notag\\
	&\geq \int_{[\delta,1-\delta]}d\eta_{j} \int_{[-\delta,\delta]^{\theta(l)-1}}
		\Bigg(\prod_{\substack{k\in\Lambda^{\#}_{l,j}\\ k\neq j}}\!\! d\eta_k \Bigg)\, F_{\Lambda_{l,j}}	
		\big( (\omega_{\Lambda_{l,j}}(\eta),\omega_{\Lambda_{l,j}^c})\big)
\end{align}
for any fixed $0<\delta<1/4$. The reason for the maximum value $1/4$ for $\delta$ will become clear
in \eqref{eq:Est5Fl} below.
The $\eta_{j}$-integral on the right-hand side of \eqref{eq:Est2Fl} will be evaluated by the Birman-Solomyak formula, see \cite{BiSo75, MR1443857, HiMu10}. To do so, we rewrite
\begin{equation}
H_L = H_{0,L} + V_{j}^c + V_j + \eta_{j} U_j =: \wtilde{H}_{L,j} + \eta_{j} U_j,
\end{equation}
as a one-parameter operator family with respect to the parameter $\eta_{j}$, where
$H_{0,L}$ is the Dirichlet restriction of $H_0$ to $\Lambda_L$ and
\begin{equation}
 U_j := \sum_{k\in \Lambda_{l,j}^{\#}} u_k, \qquad
 V_j := \sum_{\substack{k\in \Lambda_{l,j}^{\#}\\k\neq j}} \eta_k u_k, \qquad
 V_{j}^{c} := \sum_{k\not\in \Lambda_{l,j}^{\#}} \omega_k u_k\big|_{\Lambda_L}.
\end{equation}
The Birman-Solomyak formula yields
\begin{align}
	\label{eq:Est3Fl}
	\int_{[\delta,1-\delta]}d\eta_{j} \, & F_{\Lambda_{l,j}}
	\big( (\omega_{\Lambda_{l,j}}(\eta),\omega_{\Lambda_{l,j}^c})\big) \notag\\
	&= \int_{[\delta,1-\delta]}d\eta_{j} \,  \Tr\big(U_j f'_{E,\epsilon}(\wtilde{H}_{L,j}+\eta_{j} U_j)\big)
		\notag\\
	&= \Tr\left( f_{E,\epsilon}\big(\wtilde{H}_{L,j}+(1-\delta) U_j\big) -f_{E,\epsilon}(\wtilde{H}_{L,j}+\delta U_j)  \right).
\end{align}
For the values of the parameters $(\eta_k)_{k\in \Lambda_{l,j}^{\#}}$
in the integration in \eqref{eq:Est2Fl}, we have the estimate $-\delta U_{j} \le V_{j} \le \delta U_{j}$ so that
\eqref{eq:Est3Fl} implies
\begin{equation}
	\int_{[\delta,1-\delta]}d\eta_{j} \,  F_{\Lambda_{l,j}}
	\big( (\omega_{\Lambda_{l,j}}(\eta),\omega_{\Lambda_{l,j}^c})\big) 	\\
	\geq \Tr\big(f_{E,\epsilon}( H_{L,j,+})\big) -   \Tr\big(f_{E,\epsilon}( H_{L,j,-})\big).
	\label{eq:Est5Fl}
\end{equation}
Here we have introduced the operators
\begin{equation}
\label{eq:DefHLl}
\begin{aligned}
H_{L,j,+}&:= H_{0,L} + V_j^c+(1-2\delta)U_j,\\
H_{L,j,-}&:= H_{0,L} + V_j^c+2\delta U_j
\end{aligned}
\end{equation}
with Dirichlet boundary conditions on $\partial\Lambda_{L}$.
Combining \eqref{eq:Est5Fl}, \eqref{eq:Est2Fl} and \eqref{eq:Est1Fl}, we find
\begin{equation}
	\E [F_{\Lambda_{l,j}}] \geq  \frac{(2\delta\rho_-)^{\theta(l)}}{2\delta}  \; \mathbb{E}
	\big[\Tr \big( (1-f_{E,\epsilon})(H_{L,j,-}) - (1-f_{E,\epsilon})(H_{L,j,+}) \big)\big].
\end{equation}
Substituting this lower bound into \eqref{eq:DOSBound20} and subsequently taking the limit $\epsilon\searrow 0 $ in \eqref{eq:DOSBound1}, we obtain the estimate
\begin{equation}
	\label{eq:DOSBound2}
	\frac{n_L(E)}{L^d}\geq \frac{(2\delta\rho_-)^{\theta(l)}}{4C_{u,+}\delta}
	\frac{1}{L^d} \sum_{j\in \Gamma_L^l}  \mathbb{E}
	\big[ \Tr\big(\id_{(-\infty,E]}(H_{L,j,-}) -\id_{(-\infty,E]}(H_{L,j,+}) \big) \big]
\end{equation}
for Lebesgue-a.e.\ $E\in\R$.
Here, $n_L$ is the averaged finite-volume density of states of $H$, i.e. the Lebesgue density of the Lipschitz function
$E\mapsto\mathbb{E}\left[N_L(E)\right]$, see \eqref{eq:WegnerEst1}. The expectation in \eqref{eq:DOSBound2} is effectively only a partial one, $\mathbb{E}_{\Lambda_{l,j}^c}$, because no other random variables are present any more. To deduce \eqref{eq:DOSBound2}, we also used the fact that
$(1-f_{E,\epsilon})(H_{L,j,\pm})$ converges strongly to $\id_{(-\infty,E]}(H_{L,j,\pm})$ as $\epsilon\searrow 0$ and, hence, in trace class because the latter operators are of finite rank, uniformly in $\eps$. Moreover, the partial expectation $\mathbb{E}_{\Lambda_{l,j}^c}\left[\, \cdot \, \right]$ was interchanged with the limit $\varepsilon\searrow 0$ by dominated convergence because
\begin{equation}
\Tr\left( (1-f_{E,\epsilon})(H_{L,j,\pm}) \right) \leq \Tr\left( \id_{(-\infty,E+1]}(H_{L,j,\pm}) \right) \leq \Tr\left( \id_{(-\infty,E+1]}(H_{0,L})\right).
\end{equation}

\subsection{Dirichlet-Neumann bracketing.}
\label{subsec:IDOSDirNeu}

For the time being we fix an arbitrary centre $j\in \Gamma_{L}^{l}$ and an energy $E\in I \subset \FMB$.
The potential of the Schr\"odinger operators $H_{L,j,\pm}$ consists of
a deterministic part $2\delta U_j$, respectively $(1-2\delta) U_j$, which is mainly supported on $\Lambda_{l,j}$,
and of a random part $V_j^c$ mainly supported on $\Lambda_{L,l,j}:=\Lambda_L\setminus\overline{\Lambda_{l,j}}$. Leaking effects, caused by the range $R_u$ of the single-site potential $u$, may occur close to the boundary $\partial\Lambda_{l,j}$ so that both potentials may be seen simultaneously there. In the next step we will separate the two parts of the potential by a Dirichlet-Neumann bracketing argument. The arising error will be controlled by Theorem \ref{lem:SSFEstimate}.
Dirichlet-Neumann bracketing gives
\beq
\be{aligned}
H_{L,j,-} &\le \left( H_{L,j,-} \right)_{\Lambda^+_{l,j}}^D \oplus \left( H_{L,j,-} \right)_{\Lambda^+_{L,l,j}}^{D} , \\
H_{L,j,+}
 &\ge  \left( H_{L,j,+} \right)_{\Lambda^+_{l,j}}^N \oplus \left( H_{L,j,+} \right)_{\Lambda^+_{L,l,j}}^{N}.
\end{aligned}
\eeq
Here, the superscript $D$, respectively $N$, refers only to the additional Dirichlet, respectively Neumann, boundary condition along $\partial \Lambda^+_{l,j}$, where
\beq
\Lambda^+_{l,j}: = \big\{x\in \R^d:\ | x - j |_\infty < l/2 + R_u   \big\} \subset \Lambda_L
\eeq
is an enlarged version of the cube $\Lambda_{l,j}$ and $\Lambda^+_{L,l,j}: = \Lambda_L \setminus
\overline{\Lambda^+_{l,j}}$ its open complement.
Hence, the expectation in \eqref{eq:DOSBound2}  is bounded from below according to
\begin{align}
	\label{eq:DOSBound3}
	\mathbb{E} \Big[\Tr\Big( &\id_{(-\infty,E]}(H_{L,j,-}) -\id_{(-\infty,E]}(H_{L,j,+}) \Big) \Big]  \notag \\
	& \ge \mathbb{E}\Big[ \Tr\Big(\id_{(-\infty,E]}\big(\left( H_{L,j,-} \right)_{\Lambda^+_{l,j}}^D\big)
			- \id_{(-\infty,E]}\big(( H_{L,j,+} )_{\Lambda^+_{l,j}}^N\big)\Big)\Big]\notag \\
	& \quad +  \mathbb{E}\Big[ \Tr\Big(\id_{(-\infty,E]}\big(( H_{L,j,-} )_{\Lambda^+_{L,l,j}}^{D}\big)
			- \id_{(-\infty,E]}\big(( H_{L,j,+} )_{\Lambda^+_{L,l,j}}^{N}\big)\Big) \Big] .
\end{align}
The operators $(H_{L,j,\pm})^{D/N}_{\Lambda^+_{L,l,j}}$ do not see the potential $2\delta U_{j}$, respectively
$(1-2\delta)U_{j}$ any more due to the particular choice of the enlarged cube $\Lambda_{l,j}^{+}$. In fact, they are equal to the restrictions of $H$ to $\Lambda_{L} \setminus \overline{\Lambda_{l,j}^{+}}$ with Dirichlet, respectively Neumann, boundary conditions along the inner boundary $\partial\Lambda_{l,j}^{+}$ and Dirichlet boundary conditions along the outer boundary $\partial\Lambda_{L}$. By the definition of $\Gamma_{L}^{l}$ we have $\dist(\partial\Lambda_{L}, \partial\Lambda_{l,j}^{+}) \ge 3$ so that we can apply
Theorem \ref{lem:SSFEstimate} to the expectation in the last line of \eqref{eq:DOSBound3}. We conclude that this expectation is bounded by $|\partial\Lambda_{l,j}^{+}|$ times a constant which is uniform in $L > l +2R_{u}+6$, uniform in $j\in\Gamma_{L}^{l}$ and uniform in $E\in I$.
We combine this with \eqref{eq:DOSBound2} and obtain the lower bound
\begin{align}
\label{eq:DOSBound4}
\frac{n_L(E)}{L^d} \geq  \frac{(2\delta\rho_-)^{\theta(l)} b_{L,l}}{4C_{u,+}\delta}
 \l\{ \frac{1}{|\Lambda_{l}^{+}|} \mathbb{E}\left[ \Tr\Big(\id_{(-\infty,E]}\big(H_{\Lambda^+_{l}}^{D,-}\big)-\id_{(-\infty,E]}\big(H_{\Lambda^+_{l}}^{N,+}\big)\Big)\right]
- \frac{C}{l} \r\} \notag \\[.5ex]
\end{align}
for all lengths $L,l >0$ subject to $l < L - 2R_{u}-6$ with a constant $C>0$ that is uniform in $l$, $L$ and
$E\in I$. Here, we introduced the notations
$b_{L,l}:= |\Gamma_L^l| |\Lambda^+_{l,0}| / L^d$,
\beq
  H_{\Lambda^+_{l}}^{D,-} := \left( H_{L,0,-} \right)_{\Lambda^+_{l,0}}^D
\quad \text{and} \quad
  H_{\Lambda^+_{l}}^{N,+} := \left( H_{L,0,+} \right)_{\Lambda^+_{l,0}}^N
\eeq
and used translation invariance in the derivation of \eqref{eq:DOSBound4}, i.e.\ independence of $j$ of the expectation in the middle line of \eqref{eq:DOSBound3}.  For later use, we observe that, given any length $l>0$ there exists a length $\mathcal{L}(l)$ such that
\begin{equation}
	\label{calL}
  b_{L,l} \ge \frac{1}{2} \qquad \text{for every~~} L \ge \mathcal{L}(l).
\end{equation}

Now, we fix $E_0\in I$ and $0<\eps < C_{u,-}/2$ such that $[E_0-\eps, E_0+\eps]\subset I$. Then, the lower bound
\begin{multline}
	\label{locIDOSbd}
	\mathbb{E}\Big[ \Tr\Big( \id_{(-\infty,E]}\big(H_{\Lambda^+_{l}}^{D,-}\big)
		- \id_{(-\infty,E]}\big(H_{\Lambda^+_{l}}^{N,+}\big)\Big)\Big]\\
	\ge \mathbb{E}\Big[ \Tr\Big( \id_{(-\infty,E_0-\eps]}\big(H_{\Lambda^+_{l}}^{D,-}\big)
		-\id_{(-\infty,E_0+\eps]}\big(H_{\Lambda^+_{l}}^{N,+}\big)\Big)\Big]
\end{multline}
holds for all $E\in [E_0-\eps, E_0+\eps]$. We define $U:=\sum_{k\in\Z^{d}}u_{k}$ and observe the pointwise convergence
\beq
\label{periodicIDOS}
\begin{aligned}
 \lim_{l\to\infty}  \frac{1}{\Lambda_{l}^{+}} \,
 \mathbb{E} \Big[\Tr \Big(\id_{(-\infty,\tilde E]}\big(H_{\Lambda^+_{l}}^{D,-}\big)\Big)\Big]
 &= N_{H_0+ 2 \delta U}(\tilde E), \\
 \lim_{l\to\infty} \frac{1}{\Lambda_{l}^{+}} \,
 \mathbb{E}\Big[ \Tr\Big(\id_{(-\infty,\tilde E]}\big(H_{\Lambda^+_{l}}^{N,+}\big) \Big)\Big]
 &= N_{H_0+ (1-2 \delta) U}(\tilde E)
\end{aligned}
\eeq
for all $\tilde E\in \R$, where $N_A(\cdot)$ stands for the IDOS of the operator $A$.
The limits in \eqref{periodicIDOS} exist because $H_0$ and $U$ are both $\Z^{d}$-periodic and
deviations from this periodic potential -- both deterministic and random -- in the box $\Lambda_{l}^{+}$ occur only in a boundary layer whose volume scales with $l^{d-1}$.
The upper and lower covering conditions \eqref{uk-bounds} imply
\begin{align}
	\label{perIDOS-bound}
	N_{H_0+ 2 \delta U}(E_0- \eps) &- N_{H_0+ (1-2 \delta) U}(E_0 +\eps) \notag \\
	& \ge N_{H_0 + 2 \delta C_{u,+}}(E_0 - \eps) - N_{H_0+ (1-2 \delta) C_{u,-}}(E_0 +\eps) \notag \\
  & = N_0 \big(E_0 -\eps - 2 \delta C_{u,+}\big) - N_0 \big( E_0 +\eps - (1-2\delta) C_{u,-}\big)  \notag\\
  & =: K(E_{0},\varepsilon,\delta),
\end{align}
where $N_0(\cdot)$ denotes the IDOS of $H_0$.
Now, we choose
\begin{equation}
 \delta = \delta_{\varepsilon} < \frac{C_{u,-}-2\varepsilon}{2(C_{u,+} + C_{u,-})}.
\end{equation}
Here we used $\varepsilon < C_{u,-}/2$.
This choice ensures that
\begin{equation}
	E_{-} := E_0 +\eps - (1-2\delta_{\varepsilon}) C_{u,-} <  E_0-\eps - 2 \delta_{\varepsilon} C_{u,+}  =: E_{+}
\end{equation}
and, hence, that
$K(E_{0},\varepsilon) := K(E_{0},\varepsilon,\delta_{\varepsilon}) \ge0$. But we need strict positivity $K(E_{0},\varepsilon)>0$. We claim that this follows if $I \subset \Int(\Sigma_{0} + [0, C_{u,-}])$, which we require from now on in addition. Indeed, in this case, there exists $E_{0}^{0} \in \Sigma_{0}$ and $\lambda \in (0,1)$ such that $E_{0}= E_{0}^{0} + \lambda C_{u,-}$ and we have
\begin{equation}
 	E_{0}^{0} - (1-\lambda) C_{u,-} < E_{-} < E_{+} < E_{0}^{0} + \lambda C_{u,-}.
\end{equation}
We need to distinguish three cases to finish the argument for strict positivity.
\quad (i)~ $E_{0}^{0}\in (E_{-}, E_{+})$. In this case, the claim follows directly because $\Sigma_{0}$ is the set of growth points of the IDOS $N_{0}$. \quad (ii)~ $E_{0}^{0} \in [E_{+}, E_{0}^{0} + \lambda C_{u,-})$. In this case, we decrease the values of $\varepsilon$ and $\delta_{\varepsilon}$ and obtain again $E_{0}^{0}\in (E_{-}, E_{+})$ as in the first case.
\quad (iii)~ $E_{0}^{0} \in (E_{0}^{0} -(1- \lambda) C_{u,-}, E_{-}]$. Again, by making $\varepsilon$ and $\delta_{\varepsilon}$ smaller, we obtain $E_{0}^{0} \in (E_{-}, E_{+})$, and the argument is complete.

Combining \eqref{periodicIDOS} and \eqref{perIDOS-bound}, we infer
\beq
	\lim_{l\to\infty} \frac{1}{|\Lambda^+_{l}|} \, \mathbb{E}\Big[ \Tr\Big(\id_{(-\infty,E_0-\eps]}
		\big(H_{\Lambda^+_{l}}^{D,-}\big)
		- \id_{(-\infty,E_0+\eps]}\big(H_{\Lambda^+_{l}}^{N,+}\big)\Big)\Big]
	\geq K(E_{0},\varepsilon).
\eeq
This inequality, the positivity of $K(E_0,\epsilon)$  and \eqref{locIDOSbd} yield the existence of a length $l_{0}= l_{0}(E_{0},\varepsilon)$ such that for all
$l \ge l_{0}$ and all $E \in [E_{0} -\varepsilon,E_{0}+\varepsilon]$ we have
\begin{equation}
	\label{l0eins}
	\frac{1}{|\Lambda_{l}^{+}|} \, \mathbb{E}\Big[ \Tr\Big(\id_{(-\infty,E]}\big(H_{\Lambda^+_{l}}^{D,-}\big)
	- \id_{(-\infty,E]}\big(H_{\Lambda^+_{l}}^{N,+}\big)\Big)\Big]
	\ge \frac{1}{2} \, K(E_{0},\varepsilon).
\end{equation}
By possibly enlarging $l_{0}$, we also ensure
\begin{equation}
	\label{l0zwei}
  \frac{C}{l_{0}} \le \frac{1}{4} \,  K(E_{0},\varepsilon),
\end{equation}
where $C$ is the constant in \eqref{eq:DOSBound4}. We define an initial length
$L_{0} := L_{0}(E_{0},\varepsilon) := \max\big\{l_{0}+ 2R_{u} +5, \mathcal{L}(l_{0}) \big\}$ and conclude from
\eqref{l0zwei}, \eqref{l0eins}, \eqref{calL} and \eqref{eq:DOSBound4} with $l=l_{0}$ and $\delta=\delta_{\varepsilon}$ that
\begin{equation}
  \frac{n_L(E)}{L^d} \geq  \frac{(2\delta_{\varepsilon}\rho_-)^{\theta(l_{0})}}{32 C_{u,+}\delta_{\varepsilon}}
  \, K(E_{0},\varepsilon) >0
\end{equation}
for every $L \ge L_{0}$ and every  $E\in [E_0-\eps, E_0+\eps]$.
By compactness, we cover $I$ with finitely many intervals of the form $(E_0-\eps, E_0+\eps) \cap I$ and we arrive at the claimed bound after integrating over $E$ from $E_{1}$ to $E_{2}$.
\hspace*{\fill}\qedsymbol

\section{Proof of Theorem \ref{lem:SSFEstimate}}
\label{sec:SSFEst}

We fix an energy $E\in I$ and let $G:=\Lambda_L\setminus \overline{\Lambda_l(x_{0})}$ be as in the hypothesis. We then define an enlargement of G by
\begin{equation}
 G_+ := \big\{x\in \R^d:\ \dist( x, G ) < 1/2 \big\}
\end{equation}
and denote a thickened inner boundary of $G$ by
\begin{equation}
\label{eq:SSFEst1}
\partial G_- := \big\{ x\in G_+:\  \dist\big(x,\partial\Lambda_l(x_{0})\big)< 3  \big\}.
\end{equation}
We estimate the spectral shift function of the pair $(H_{L,l}^D, H_{L,l}^N)$ by
\begin{align}
\label{eq:SSFEst2}
0 \le \xi(E,H_{L,l}^{N},H_{L,l}^{D}) &=
\Tr \Big( \id_{(-\infty,E]}( H_{L,l}^N) - \id_{(-\infty,E]}( H_{L,l}^D)\Big) \notag \\
&\leq
\sum_{x\in\partial G_-^{\#}}\Big\Vert\chi_{x}\Big(\id_{(-\infty,E]}(H_{L,l}^D)-\id_{(-\infty,E]}(H_{L,l}^N)\Big)\chi_{x}\Big\Vert_1 \notag \\
& \quad +
\sum_{x\in \partial G_-^{\#,c}} \Big\Vert\chi_{x}\Big(\id_{(-\infty,E]}(H_{L,l}^D)-\id_{(-\infty,E]}(H_{L,l}^N)\Big)\chi_{x}\Big\Vert_1 \notag \\
& =:  I_1+I_2,
\end{align}
where $\| \pmb\cdot \|_p$ denotes the Schatten-$p$-norm and, recalling the notation \eqref{eq:def-raute}, we have set $\partial G_-^{\#,c}:= G_+^{\#} \setminus \partial G_-^{\#}$.
In the following we treat the two contributions $I_1$ and $I_2$ separately. For $I_{1}$ we apply the a priori estimate from Lemma \ref{lem:aPriori} with $p=1$. It gives a non-random constant $C_1$, uniform in $E \in I$, $L$  and $l$, such that
\begin{equation}
\label{eq:SSFEst3}
\Vert\chi_{x}\id_{(-\infty,E]}(H_{L,l}^{D/N})\chi_x\Vert_1 \leq C_1
\end{equation}
holds almost surely for all $x\in\Rn$.
Hence,
\begin{equation}
\label{eq:SSFEst4}
\E[I_1] \leq 2C_{1} | \partial G_-^{\#} | \leq C l^{d-1}
\end{equation}
with a constant $C>0$ uniform in $E \in I$ and $L,l$.

For $I_2$ we apply the inequality $\|A\|_1 \leq \|A\|^{1/2}_{1/2} \|A\|^{1/2}$, which follows from $\|A\|_1 = \sum _{j}\mu_j(A)$ and the bound $\mu_j(A) \leq \|A\|$, for singular values $\mu_j(A)$, and the a priori estimates from Lemma \ref{lem:aPriori} with $p=1/2$. This yields almost surely the upper bound
\begin{equation}
\label{eq:SSFEst5}
I_2 \leq C_{1/2}^{1/2} \sum_{x\in \partial G_-^{\#,c}}\Big\Vert\chi_{x} \Big(\id_{(-\infty,E]}(H_{L,l}^D)-\id_{(-\infty,E]}(H_{L,l}^N)\Big)\chi_{x}\Big\Vert^{1/2}
\end{equation}
with a non-random constant $C_{1/2}$ that is uniform in $E \in I$, $L$  and $l$.
As in the Appendix we set $E_{0} := \inf_{x\in\Rn} V_{0}(x)$ so that $H_{L,l}^{D/N}\geq E_0$ holds.
Since $E$ is not an eigenvalue of $H_{L,l}^{D/N}$ almost surely, we represent the Fermi projections by an integration of the resolvents along a closed rectangular contour $\mathcal{C}_E$ in the complex plane connecting the points $E+i$, $E_0-1+i$, $E_0-1-i$ and $E-i$ by straight line segments. We substitute these representations into \eqref{eq:SSFEst5}
and bound the norm of the integral by an integral of the norm.
For any $s \in (0, 1)$, we factor the norm into the product of two powers $s/2$ and $1-s/2$ of the norm, and bound the latter using the triangle inequality
and basic estimate $\|R_z(H_{L,l}^{D/N})\|^{1-s/2} \le |\Im z|^{-(1-s/2)}$. Altogether, this results in the bound
\begin{equation}
\label{eq:SSFEst6}
I_2 \leq (2C_{1/2})^{1/2}
 \sum_{x\in \partial G_-^{\#,c}} \bigg(\int_{\mathcal{C}_E} \frac{|dz|}{|\Im(z)|^{1-s/2}} \, \Big\|\chi_x \Big(R_z(H_{L,l}^D)-R_z(H_{L,l}^N)\Big) \chi_x \Big\|^{s/2}\bigg)^{1/2},
\end{equation}
where the notation $\int_{\mathcal{C}_E} |dz|$ stands for the sum of the absolute values of the four complex line integrals which make up the contour $\mathcal{C}_E$.

Next, we apply the geometric resolvent equation to the norm in \eqref{eq:SSFEst6}. To do so, we choose a switch function $\psi \in C^{2}(G)$ with $\dist\big(\supp(\psi),\partial \Lambda_l(x_{0})\big)\geq 1/4$,
\begin{equation}
 	\supp(\nabla \psi)\subseteq \Big\{x\in G: 1/4 \leq \dist\big(x,\partial \Lambda_l(x_{0})\big)\leq 1/2\Big\}
	=:\Omega,
\end{equation}
$ \|\nabla \psi\|_{\infty} \leq 8$ and $1 \geq \psi\geq \id_{G\setminus \partial G_-}$.
In analogy to the definition of $G_+$, we introduce the enlarged set $\Omega_+:= \{ x\in \R^d:\ \dist( x, \Omega ) < 1/2 \}$ and conclude
\begin{align}
\Big\|\chi_x \Big(R_z(H_{L,l}^D) &- R_z(H_{L,l}^N)\Big) \chi_x \Big\| \notag \\
&=
\big\| \chi_x R_z(H_{L,l}^D)[-\Delta,\psi]R_z(H_{L,l}^N)\chi_x\big\| \notag \\
&\leq \sum_{y\in \Omega_+^{\#}} \big\| \chi_x R_z(H_{L,l}^D)\chi_y\big\| \big\|\chi_y [-\Delta,\psi]R_z(H_{L,l}^N)\chi_x\big\|	\label{eq:SSFEst7}
\end{align}
for every $x\in \partial G_-^{\#,c}$. Here, the operator $\psi H_{L,l}^N - H_{L,l}^D \psi = -[-\Delta,\psi]$
is a differential operator of order one acting only on $\supp(\nabla \psi)$ and, hence, insensitive to any boundary condition of the involved Laplacian.
Since $\dist\left( \Omega_+^{\#},\partial G_-^{\#,c} \right) \geq 2$, we have $\dist \big(\Lambda_{1}(x),{\Lambda_2(y)}\big)\geq 1/2$ for all
$x\in  \partial G_-^{\#,c}$ and all $y\in \Omega_+^{\#}$. Hence the norms involving $[-\Delta,\psi]$ on the right hand side of \eqref{eq:SSFEst7} can be estimated in a standard manner, see, for example,
     \cite[Lemma 2.5.3]{MR1935594} and the proof of \cite[Lemma 2.5.2]{MR1935594}. This yields a constant $c$, which is uniform in $E\in I$, $L$ and $l$, such that
\begin{equation}
\label{eq:SSFEstResBound}
\big\|\chi_y [-\Delta,\psi]R_z(H_{L,l}^N)\chi_x\big\| \leq
c\, \big\|\id_{\Lambda_2(y)} R_z(H_{L,l}^N)\chi_x\big\|.
\end{equation}
Combining \eqref{eq:SSFEst6}, \eqref{eq:SSFEst7} and \eqref{eq:SSFEstResBound}, we get
\begin{equation}
\label{eq:SSFEst9}
I_2 \leq  c' \sum_{\substack{y\in\Omega_+^{\#} \\ x\in\partial G_-^{\#,c}}}
\bigg( \int_{\mathcal{C}_E} \frac{|dz|}{|\Im(z)|^{1-s/2}} \, \|\chi_x R_z(H_{L,l}^D) \chi_{y}\|^{s/2} \|\id_{\Lambda_2(y)} R_z(H_{L,l}^N)\chi_x\|^{s/2}\bigg)^{1/2}
\end{equation}
with $c':= c^{s/4}(2C_{1/2})^{1/2}$. Next, we take the expectation on both sides of \eqref{eq:SSFEst9}, apply Jensen's inequality to the concave square root function, and the Cauchy-Schwarz inequality to the expectation and
exploit exponential localisation \eqref{eq:DefFMB} of fractional moment bounds. We find that there exist finite constants $C,C',C''>0$, all uniform in $E\in I$, $L$ and $l$, such that
\begin{align}
	\label{eq:SSFEst10}
	\mathbb{E} \left[I_2\right] &\leq  C \sum_{\substack{y\in\Omega_+^{\#} \\ x\in\partial G_-^{\#,c}}} e^{-\mu|x-y|}
		\bigg(\int_{\mathcal{C}_E} \frac{|dz|}{|\Im(z)|^{1-s/2}} \;
		\mathbb{E}\left[ \|\id_{\Lambda_2(y)}R_z(H_{L,l}^N) \chi_x\big\|^s  \right]^{1/2}\bigg)^{1/2} \notag\\
	&\leq C' \sum_{y\in\Omega_+^{\#}} \; \sum_{x\in\Zd} e^{-\mu |x-y|} \notag \\
	&\leq C''\, l^{d-1}.
\end{align}
For the second inequality in \eqref{eq:SSFEst10} we covered $\Lambda_2(y)$ by $2^d$ boxes of side-length $1$ and used the a priori bound
\begin{equation}
\label{eq:SSFEstAPri}
\sup_{{x,y\in\Rn\!,\, E'\in I,\, \eta\neq 0}}\mathbb{E}\left[ \|\chi_y R_{E'+i\eta}(H_{L,l}^N) \chi_x\big\|^s  \right] \leq \wtilde C  < \infty
\end{equation}
with a constant $\wtilde C$ that does not depend on $L$ or $l$. Its validity follows from
\cite[Lemma~3.3]{artRSO2006AizEtAl2}, see also \cite[Lemma~4]{MR2303305}. The bound is stated there for operators with Dirichlet boundary conditions, but it generalises to mixed Dirichlet and Neumann boundary conditions, as needed for \eqref{eq:SSFEstAPri}. To see that such a priori bounds are insensitive to the boundary condition we note that their proofs rely on two-parameter spectral averaging for the resolvent $R_{z}(A)$ of a maximally dissipative operator $A$ and $\Im(z) >0$,  see Lemma~3.1 and Appendix~C in \cite{artRSO2006AizEtAl2} or
Lemma~3 and Appendix~A.3 in \cite{MR2303305}. Finally, if \eqref{eq:SSFEstAPri} holds for $z=E'+i\eta$ with $\eta>0$, then it also holds with $\eta<0$ by taking the adjoint.
\hspace*{\fill} \qedsymbol

\be{appendix}

\section{A priori bounds: Supertrace-class conditions}\label{sec:apriori}

The deterministic Lemma \ref{lem:aPriori} below is essential for the proof of Theorem \ref{lem:SSFEstimate}. We prove it here for completeness and convenience of the reader. It is known for operators with Dirichlet boundary conditions, see \cite[App. A]{artRSO2006AizEtAl2} and \cite{MR2303305}. We closely follow the approach in \cite{MR2303305}.

We consider the following deterministic Schr\"odinger operator
\begin{MyDescription}
\item[(D)] {$H:=-\Delta+V_0+ V$ with two bounded potentials $V_0, V \in L^{\infty}(\R^{d})$ such that $0\le V\le M$ for some finite constant $M>0$.}
\end{MyDescription}
Let $E_{0} := \inf_{x\in\Rn} V_{0}(x)$ so that $H\geq E_0$.
As before, we denote by $H_{L,l}^{D}$, respectively $H_{L,l}^{N}$, the restriction of $H$ to $\Lambda_L\setminus\overline{\Lambda_l(x_{0})}$ with Dirichlet, respectively Neumann, boundary conditions along the inner boundary $\partial\Lambda_l(x_{0})$ and Dirichlet boundary conditions along the outer boundary $\partial\Lambda_L$. As we have to deal with Schatten-$p$ classes for $0<p\le1$ we note that the (generalised) H\"older inequality for Schatten classes remains true for H\"older exponents $p_1,...,p_n>0$ subject to $p_1^{-1}+...+p^{-1}_{n}=p^{-1}$. Moreover the ``triangle-like''
inequality
\beq
\label{ptriangle}
\|A+B\|_p^p\leq \|A\|_p^p+\|B\|_p^p
\eeq
holds for compact operators $A,B$ and $p\in (0,1]$, see \cite[Thm. 2.8]{McCarthy}.

\begin{lemma}
	\label{lem:aPriori}
Assume {\upshape (D)}. Let $p>0$, $I\subset \R$ compact and $M>0$ fixed. Then there exists a finite constant $C_p$, which depends on $I$ only through $\max I$, such that
	for $\star\in \{ D,N \}$, for all $x,y\in\mathbb{R}^{d}$, for all measurable $g : \R \rightarrow \C$ with $|g|\leq 1$ and $\supp(g) \subseteq I$,
	for all $L,l>0$ and
	$x_{0}\in\Lambda_{L}$ such that $\overline{\Lambda_{l}(x_{0})} \subset \Lambda_{L}$
	and for all measurable potentials $V: \Rn \rightarrow [0,M]$ we have the estimate
	\begin{equation}
	\label{eq:aPrioriStat}
	\|\chi_x g(H^{\star}_{L,l})\chi_y\|_p \leq C_{p}.
	\end{equation}
\end{lemma}

The above Lemma follows, up to some iteration procedure, from the following Schatten-class Combes-Thomas estimate.

\begin{lemma}
	\label{lem:SchattenCombesThomas}
	Assume {\upshape (D)}.
	Let $p>d/2$, $E \in (-\infty,E_0)$ and $M>0$. Then, there exist finite constants
	$C_{p,E},\mu_{p,E}>0$ such that for all $x,y\in\mathbb{R}^d$, for all $L,l>0$ and
	$x_{0}\in\Lambda_{L}$ such that $\overline{\Lambda_{l}(x_{0})} \subset \Lambda_{L}$ and for all
	measurable potentials $V: \Rn \rightarrow [0,M]$ we have the estimate
	\begin{equation}
	\label{eq:SchattenCTStat}
	\|\chi_x (H^{\star}_{L,l}-E)^{-1}\chi_y \|_p \leq C_{p,E} \, e^{-\mu_{p,E}|x-y|}.
	\end{equation}	
\end{lemma}

\begin{proof}
	Let $E\in (-\infty,E_0)$, set $G:= \Lambda_L\setminus\overline{\Lambda_l(x_{0})}$ for fixed $L,l>0$ and $x_{0}\in\Lambda_{L}$ such that $\overline{\Lambda_{l}(x_{0})} \subset \Lambda_{L}$. As before, let 	
	$G_+^\#:=\{n\in\Zd: \dist(n,G) < 1/2\}$. For $n\in G_+^\#$ we introduce the rectangular box $Q_n:= \Lambda_{1}(n)\cap G$ and the Neumann Laplacian $-\Delta^N_{Q_n}$  on $Q_{n}$.
	Dirichlet-Neumann bracketing
	\begin{equation}
	H^{D}_{L,l}\geq H^{N}_{L,l} \geq -\Delta^{N}_{L,l} + E_{0} \geq \bigoplus_{n\in G_+^\#} \Big(
		-\Delta^N_{Q_n} +E_0 \Big)
	\end{equation}
	and \cite[(2.21) in Sec.\ VI.2]{MR1335452} then imply the bound
	\beq
	\begin{aligned}
	&\Big\Vert\Big(\bigoplus_{n\in G_+^\#} (-\Delta^N_{Q_n} + E_0 - E)^{1/2}\Big)(H_{L,l}^{\star}- E)^{-1/2}\Big\Vert \leq 1.
	\end{aligned}
	\label{eq:aPriNB3}
	\eeq
	We set $\wtilde E:=E_0 - E >0$.
	Hence, using H\"older's inequality, we estimate for $x\in \Rn$ and fixed $p' \ge 1$ (to be determined later)
	\begin{align}
	\label{eq:aPriNB4}
	\|\chi_x(H^{\star}_{L,l}-E)^{-1/2}\|_{p'}
	&\leq  \Big\|\chi_x \Big(\bigoplus_{n\in G_+^\#} (-\Delta^N_{Q_n} + \wtilde E)^{-1/2}\Big)\Big\|_{p'} \notag \\
	&\le \sum_{n\in G_+^\#}\big\|\chi_x(-\Delta^N_{Q_n} + \wtilde E)^{-1/2}\big\|_{p'}.
	\end{align}
	Since the cardinality of $\{n\in\Z^d: Q_n \cap \Lambda_{1}(x) \neq \emptyset\}$
	is at most $2d$, we conclude
	\begin{equation}
	\begin{aligned}
	\eqref{eq:aPriNB4}
	&\leq 2d \max_{\substack{n\in\Z^d:\\ Q_n\cap \Lambda_{1}(x)\neq \emptyset}}\big\|(-\Delta^N_{Q_n} +\wtilde E)^{-1/2}\big\|_{p'}.
	\end{aligned}	
	\label{eq:aPriNB40}
	\end{equation}
	For any rectangular box $\wtilde\Lambda:= \bigtimes_{j=1}^{d} (-L_j/2,L_j/2)$ with side-lengths $L_j>0$ for $1\le j\le d$, the eigenvalues of
	 $-\Delta_{\wtilde \Lambda}^N$ are given by $E_k(\wtilde\Lambda):=\sum_{j=1}^d \l(\frac{\pi k_j}{L_j}\r)^2$ and indexed by $k:=(k_1,...k_d)\in \N_0^d$. Since $Q_n$ is a rectangular box of the above form with $L_j\le 1$,
	 the eigenvalues $E_k(\wtilde\Lambda)$ are monotone decreasing in the side-lengths $L_{j}$ and the Neumann Laplacian is translation invariant, we infer from \eqref{eq:aPriNB40} that
	\be{align}
	\label{eq:SchattenCT400}
	\|\chi_x(H^{\star}_{L,l}-E)^{-1/2}\|^{p'}_{p'}
	 &\leq  (2d)^{p'}
	 \big\|(-\Delta^N_{\Lambda_{1}} +\wtilde E)^{-1/2}\big\|^{p'}_{p'} \notag \\
	 &= (2d)^{p'} \sum_{k\in\N_0^d} \big(E_k(\Lambda_{1})+ \wtilde E\big)^{-p'/2} =: C_{p'}.
	 \end{align}
	The constant $C_{p'}$ is finite for $p'>d$.	

Now, let $p>d/2$ and fix $\theta\in(0,p- d/2)$, whence $p':=2(p-\theta)>d$. We estimate
\begin{align}
\label{eq:SchattenCT4}
\big\|\chi_x(H^{\star}_{L,l}-E)^{-1}\chi_y\big\|_p^{p}
&\leq
\big\|\chi_x(H^{\star}_{L,l}-E)^{-1}\chi_y \big\|_{p-\theta}^{p-\theta}  \, \big\|\chi_x(H^{\star}_{L,l}-E)^{-1}\chi_y\big\|^{\theta} \notag \\
& \le \| \chi_x(H^{\star}_{L,l}-E)^{-1/2}\|_{2(p-\theta)}^{p-\theta}
\|(H^{\star}_{L,l}-E)^{-1/2}\chi_y\|_{2(p-\theta)}^{p-\theta} \notag\\
& \quad \times
	\big\|\chi_x(H^{\star}_{L,l}-E)^{-1}\chi_y\big\|^{\theta} \notag \\
& \le C_{2(p-\theta)} \big\|\chi_x(H^{\star}_{L,l}-E)^{-1}\chi_y\big\|^{\theta},
\end{align}
where we used \eqref{eq:SchattenCT400} in the last step.
The Combes-Thomas estimate for operator norms (e.g., \cite[Thm.~1]{MR1937430},  \cite[Thm.\ 5.4.1]{MR1935594}), which also applies to Schr\"odinger operators with mixed boundary conditions completes the proof.
\end{proof}

\begin{proof}[Proof of Lemma \ref{lem:aPriori}]
We use the abbreviation $H:=H^{\star}_{L,l}$. Without loss of generality we assume $0<p \le 1$ (because $\|\pmb\cdot\|_{p} \le \|\pmb\cdot\|_{1}$ for $p \ge 1$).
Let $m\in\N$ such that $m> d/(2p)$ and observe $H \ge E_{0}$.
We insert the $m$-th power of the resolvent on the l.h.s.\ of \eqref{eq:aPrioriStat} and estimate using H\"older's inequality
\begin{align}
\label{eq:aPriori1}
\big\|\chi_x g(H)\chi_y\big\|_p^{p}
&=
\| \chi_x g(H)(H-E_0+1)^m(H-E_0+ 1)^{-m}\chi_y\|_p^p \notag \\
&\leq
\| \chi_x g(H)(H-E_0+ 1)^{m}\|^p \;\| (H-E_0+1)^{-m}\chi_y\|_p^p \notag\\
& \le C \| (H-E_0+ 1)^{-m} \chi_y\|_p^p.
\end{align}
The last estimate holds true because $ g(H) (H-E_0+1)^m$ is a self-adjoint operator with operator norm bounded by $(\sup I-E_0+1 )^m =: C^{1/p}$. Next we set $y_{m+1}:=y$ and estimate with \eqref{ptriangle}
\begin{align}
\label{eq:aPriori2}
\| (H-E_0+1)^{-m}\chi_y\|_p^p
&= \Big\| \Big(\sum_{y_1\in\Zd} \chi_{y_1}\Big) (H-E_0+1)^{-1}  \cdots
\notag \\
& \hspace{1cm}\times \cdots \Big(\sum_{y_m\in\Zd} \chi_{y_m}\Big) (H-E_0+1)^{-1}  \chi_{y_{m+1}}   \Big\|_p^p \notag\\
&\leq \sum_{y_1,...,y_m\in \Zd}  \Big\| \prod_{l=1}^{m} \Big(\chi_{y_{l}}(H-E_0+1)^{-1}\chi_{y_{l+1}} \Big) \Big\|_p^p.
\end{align}
Using H\"older's inequality for Schatten-$p$ classes, we obtain the inequality
\begin{align}
 \Big\| \prod_{l=1}^{m} \Big(\chi_{y_{l}}(H-E_0+1)^{-1}\chi_{y_{l+1}}\Big) \Big\|_p^p
\leq &
 \prod_{l=1}^{m} \big\Vert \chi_{y_{l}} (H-E_0+1)^{-1}\chi_{y_{l+1}} \big\Vert_{pm}^{p} \notag \\
 \leq & C'  \prod_{l=1}^{m} e^{-\mu|y_{l}-y_{l+1}|},
\end{align}
where the last estimate is due to Lemma \ref{lem:SchattenCombesThomas} applied with $E=E_{0}-1$. Inserting this into \eqref{eq:aPriori2} and repeatedly using
\begin{equation}
\label{eq:aPriori3}
\sum_{y_2\in\Zd} e^{-\mu|y_{1}-y_{2}|} e^{-\mu|y_{2}-y_{3}|} \leq C'' e^{-\mu/2 |y_{1}-y_{3}|},
\end{equation}
we conclude from \eqref{eq:aPriori1} that
\begin{equation}
\label{eq:aPriori4}
\big\|\chi_x g(H)\chi_y\big\|_p^{p} \leq C_p
\end{equation}
with a constant $C_p$ which is independent of all the parameters stated in the lemma.
\end{proof}

\e{appendix}	
	
\section*{Acknowledgements}
M.G. is grateful to Gian Michele Graf for his kind hospitality at {ETH} {Z\"urich}.


\newcommand{\etalchar}[1]{$^{#1}$}

\end{document}